\documentclass[12pt,letter]{article}

\usepackage{graphicx,colordvi,psfrag,bbm}
\usepackage{amsmath,amssymb}
\usepackage{epstopdf}
\usepackage[caption=false]{subfig}
\usepackage{epsfig,cite}
\usepackage{calc
, pgf, xcolor}
\usepackage{nicefrac}
\usepackage{enumerate}
\usepackage{url}

\allowdisplaybreaks

\newtheorem{theorem}{Theorem}
\newtheorem{corollary}{Corollary}
\newtheorem{remark}{Remark}

\newtheorem{proposition}{Proposition}
\newtheorem{lemma}{Lemma}
\newtheorem{conjecture}{Conjecture}
\newenvironment{proof}[1][Proof]{\noindent\textbf{#1.} }{\ \rule{0.5em}{0.5em}}

\newcommand{\RR}{\mathbb{R}}

\def\Expt{\mathbb{E}}

\def\Expt{\mathbb{E}}

\begin{document}

\title{An Improved Upper Bound for the Most Informative Boolean Function Conjecture}
\author{Or Ordentlich, Ofer Shayevitz and Omri Weinstein \thanks{Or Ordentlich and Ofer Shayevitz are with the Department of Electrical Engineering - Systems at the Tel Aviv University, \{ordent,ofersha\}@eng.tau.ac.il. Omri Weinstein is with the Department of Computer Science, Princeton University, oweinste@cs.princeton.edu.
The work of O. Ordentlich was supported by the Admas Fellowship Program of the Israel Academy of Science and Humanities. The work of O. Shayevitz was supported by an ERC grant no. 639573, a CIG grant no. 631983, and an ISF grant no. 1367/14. The work of O. Weinstein was supported by a Simons Fellowship award in TCS and a Siebel scholarship.}}

\parskip 3pt

\date{}

\maketitle

\begin{abstract}
Suppose $X$ is a uniformly distributed $n$-dimensional binary vector and $Y$ is obtained by passing $X$ through a binary symmetric channel with crossover probability $\alpha$. A recent conjecture by Courtade and Kumar postulates that $I(f(X);Y)\leq 1-h(\alpha)$ for any Boolean function $f$. So far, the best known upper bound was $I(f(X);Y)\leq (1-2\alpha)^2$.
In this paper, we derive a new upper bound that holds for all balanced functions, and improves upon the best known bound for all $\tfrac{1}{3}<\alpha<\tfrac{1}{2}$.
\end{abstract}

\section{Introduction}\label{sec:intro}

Let $X$ be an $n$-dimensional binary vector uniformly distributed over $\{0,1\}^n$, and $Y$ be the output of passing each component of $X$ through a binary symmetric channel with crossover probability $\alpha\leq1/2$. The following was recently conjectured by Courtade and Kumar~\cite{ck14}.
\begin{conjecture}\label{conj:ck}
For any Boolean function $f:\{0,1\}^n\mapsto\{0,1\}$ it holds that
\begin{align}
I\left(f(X);Y\right)\leq 1-h(\alpha),\label{eq:Conj}
\end{align}
where $I(f(X);Y)$ is the mutual information between $f(X)$ and $Y$, and $h(p)\triangleq -p\log{p}-(1-p)\log(1-p)$ is the binary entropy function.\footnote{All logarithms are taken to base $2$.}
\end{conjecture}

For a dictatorship function, $f(X)=X_i$ the conjectured upper bound~\eqref{eq:Conj} is attained with equality. Therefore, the conjecture can be interpreted as
postulating that dictatorship is the most ``informative'' Boolean function, i.e., it achieves the maximal $I(f(X);Y)$.

So far, the best known bound that holds universally for all Boolean functions is
\begin{align}
I\left(f(X);Y\right)\leq (1-2\alpha)^2.\label{eq:gerber}
\end{align}
This bound can be established through various techniques, including an application of Mrs. Gerber's Lemma~\cite{wz73,ErkipPhD,ct14}, the strong data-processing inequality~\cite{ag76,agkn13} and standard Fourier analysis as described below. 

In this paper, we derive an upper bound on $I(f(X);Y)$ that holds for all balanced functions, and improves upon~\eqref{eq:gerber} for all $\tfrac{1}{3}<\alpha<\tfrac{1}{2}$. Specifically, we obtain the following result.


\begin{theorem}
\label{thm:main}
For any balanced Boolean function $f(x):\{0,1\}^n\mapsto\{0,1\}$, and any
$\tfrac{1}{2}\left( 1-\tfrac{1}{\sqrt{3}}\right)\leq\alpha\leq\tfrac{1}{2}$, we have that
\begin{align}
I(f(X);Y)&\leq \frac{\log(e)}{2}(1-2\alpha)^{2}+9\left(1-\frac{\log(e)}{2}\right)(1-2\alpha)^{4}.\label{eq:mainbound}
\end{align}
\end{theorem}

%
%



For the proof of Theorem~\ref{thm:main}, we first lower bound the conditional entropy $H(f(X)|Y)$ in terms of the second and fourth moments of
the random variable $(1-2P_Y^f)$, where $P_y^f\triangleq\Pr\left(f(X)=0|Y=y\right)$. Specifically, we show that
\begin{align}
H(f(X)|Y)\geq 1-\frac{\log(e)}{2}\Expt(1-2P_Y^f)^{2}+\left(1-\frac{\log(e)}{2}\right)\Expt(1-2P_Y^f)^{4}.\label{eq:entboundIntro}
\end{align}
To upper bound the second and fourth moments in \eqref{eq:entboundIntro}, we use basic Fourier analysis of Boolean functions along with a simple application
of the Hypercontractivity Theorem~\cite{ODonnellBook,Bonami1970,Beckner75}, in order to derive universal upper bounds on $\Expt(1-2P_Y^f)^{2k}$ that hold
for all balanced Boolean functions.\footnote{
We remark that using similar techniques it is possible to obtain an upper bound on $I(f(X);Y)$ of a similar form that holds for any $q$-biased function (i.e., any
function for which $\Expt(f(X))=1-2q$). However, the obtained bound is not maximized at $q=\tfrac{1}{2}$, and therefore cannot be used to establish~\eqref{eq:mainbound} for all Boolean functions.} In particular, these bounds show that $\Expt(1-2P_Y^f)^2\leq(1-2\alpha)^2$ and
$\Expt(1-2P_Y^f)^4\leq 9(1-2\alpha)^4$. Plugging these bounds in~\eqref{eq:entboundIntro} yields the Theorem.

An appealing feature of the new upper bound, is that the ratio between the RHS of~\eqref{eq:mainbound} and $1-h(\alpha)$ approaches $1$ in the limit of $\alpha\rightarrow\tfrac{1}{2}$. For the bound~\eqref{eq:gerber}, on the other hand, the same ratio does not approach $1$. Rather
\begin{align}
\lim_{\alpha\rightarrow\tfrac{1}{2}}\frac{(1-2\alpha)^2}{1-h(\alpha)}= \frac{2}{\log(e)}\approx 1.3863.\nonumber
\end{align}
%
%


\section{Preliminaries}
\label{sec:pre}

To prove our results, it will be more convenient to map the additive group $\{0,1\}^n$ to the (isomorphic) multiplicative group $\{-1,1\}^n$. Specifically, let $X$ be an $n$-dimensional binary vector uniformly distributed over $\{-1,1\}^n$ and $Y$ be the output of passing each component of $X$ through a binary symmetric channel with crossover probability $\alpha\leq1/2$. Thus, for all $i\in\{1,\cdots,n\}$ we have $Y_i=X_i\cdot Z_i$, where $\{Z_i\}_{i=1}^n$ is an i.i.d. sequence of binary random variables statistically independent of $\{X_i\}_{i=1}^n$, with $\Pr(Z_i=-1)=\alpha$ and $\Pr(Z_i=1)=1-\alpha$. Note that $Y$ is also uniformly distributed over $\{-1,1\}^n$.

Let $f:\{-1,1\}^n\mapsto\{-1,1\}$ be a \emph{balanced} Boolean function, i.e.,
\begin{align}
\Pr\left(f(X)=-1\right)=\frac{1}{2}.\nonumber
\end{align}
Note that this condition is equivalent to $\mathbb{E}_X(f(X))=0$. For each $y\in\{-1,1\}^n$ define the posterior distribution of $f(X)$ given the observation $y$
\begin{align}
P_y^{f}\triangleq \Pr(f(X)=-1|Y=y)\nonumber,
\end{align}
and note that
\begin{align}\label{eq:condE}
\mathbb{E}(f(X)|Y=y)=1-2 P_y^f.
\end{align}

In what follows, we will extensively use  Fourier  analysis of real functions on the hypercube $\{-1,1\}^n$. 
Let $X$ be a random vector uniformly distributed on $\{-1,1\}^n$ and define $[n]\triangleq \{1,2,\ldots,n\}$.
The Fourier-Walsh transform of a function $f:\{-1,1\}^n\mapsto\RR$ is given by
\begin{align}
f(x)=\sum_{S\subseteq [n]} \hat{f}(S)\prod_{i\in S} x_i,
\end{align}
where
\begin{align}
\hat{f}(S)=\mathbb{E}_X\left(f(X)\prod_{i\in S} X_i\right)
\end{align}
is the correlation of $f$ with the ``parity" function on the subset $S$. 
It is easy to verify that the basis $\left\{\varphi_S(x)=\prod_{i\in S} x_i\right\}_{S\subseteq [n]}$ is orthonormal with respect to the inner product
$<f,g>=\mathbb{E}\left(f(X)g(X)\right)$, which implies that for any two functions $f,g:\{-1,1\}^n\mapsto\RR$ it holds that
\begin{align}
\mathbb{E}\left(f(X)g(X)\right)=\sum_{S\subseteq [n]}\hat{f}(S)\hat{g}(S).
\end{align}
In particular, Parseval's identity gives $\sum_S \hat{f}^2(S)=\mathbb{E}\left(f^2(X)\right)$. Thus, if $f$ is a Boolean function, i.e., $f:\{-1,1\}^n\mapsto\{-1,1\}$,
we have $\sum_S \hat{f}^2(S)=1$.

The following proposition gives three simple properties of the Fourier coefficients that will be useful in the sequel.
\begin{proposition}[Basic properties of the Fourier transform] \label{prop:fourier_properties}
Let $f:\{-1,1\}^n\mapsto\{-1,1\}$ be a Boolean function. We have that
\begin{enumerate}[(i)]
\item $f$ is balanced if and only if $\hat{f}(\emptyset)=0$;
\item If $|\hat{f}(S)|>0$ then $|\hat{f}(S)|\geq 2^{-n}$;
\item Any balanced function $f$ which is not a dictatorship function must satisfy
$\sum_{S : \; |S|\geq 2}\hat{f}^2(S) > 0$;
\end{enumerate}
\end{proposition}

\begin{proof}
\begin{enumerate}[(i)]
\item By definition $\hat{f}(\emptyset)=\Expt f(X)=0$ for $f$ balanced.
\item We note that the sum $\sum_{x\in\{-1,1\}^n} f(x)\prod_{i\in S} x_i$ is an integer, and the claim follows immediately.
\item Let $f$ be a balanced Boolean function ($\hat{f}(\emptyset)=0$) with $\sum_{S : \; |S| \geq 2}\hat{f}^2(S) = 0$.
Then $f(X)=\sum_{i\in [n]} \hat{f}(\{i\}) X_i$. Since there always exists some $x\in\{-1,1\}^n$ for which $f(x)=\sum_i|\hat{f}(\{i\})|$, we must have that $\sum_i|\hat{f}(\{i\})|=1$. On the other hand, by Parseval's identity we have $\sum_{i=1}\hat{f}^2(\{i\})=1$. Clearly, the last two equations can simultaneously hold if and only if there is a unique $i$ for which $|\hat{f}(\{i\})|>0$. Hence, $f$ must be a dictatorship.
\end{enumerate}
\end{proof}

For any function $f:\{-1,1\}^n\mapsto\RR$ with Fourier coefficients $\{\hat{f}(S)\}$ and $\rho\in\RR^+$, the \emph{noise operator}
$T_\rho f :\{-1,1\}^n\mapsto\RR$ is defined as~\cite{ODonnellBook}
\begin{align}
(T_\rho f)(x)\triangleq\sum_{S\subseteq[n]} \hat{f}(S)\rho^{|S|}\prod_{i\in S} x_i.
\end{align}
Recall that in our setting $X$ and $Y$ are the input and output of a binary symmetric channel with
crossover probability $\alpha$. Thus we can use the Fourier representation of $f$ to write $\mathbb{E}\left(f(X)|Y=y\right)$ as follows~\cite{ODonnellBook}
\begin{align}
\mathbb{E}\left(f(X)|Y=y\right)&=\mathbb{E}\left(\sum_{S\subseteq[n]} \hat{f}(S)\prod_{i\in S} X_i \ \big|Y=y \right)\nonumber\\
&=\sum_{S\subseteq[n]} \hat{f}(S)\mathbb{E}\left(\prod_{i\in S} y_i Z_i\right)\nonumber\\
&=\sum_{S\subseteq[n]} \hat{f}(S)\mathbb{E}\left(\prod_{i\in S} Z_i\right)\prod_{i\in S} y_i\nonumber\\
&=\sum_{S\subseteq[n]} \hat{f}(S)\prod_{i\in S}\mathbb{E}\left(Z_i\right)\prod_{i\in S} y_i\nonumber\\
&=\sum_{S\subseteq[n]} \hat{f}(S)(1-2\alpha)^{|S|}\prod_{i\in S} y_i,\label{eq:FourierGeneral}\\
&=(T_{1-2\alpha}f)(y),
\end{align}
where we have used the fact that $\{Z_i\}$ is an i.i.d. sequence. Recalling equation~\eqref{eq:condE}, this also yields
\begin{align}
1-2 P_y^f=(T_{1-2\alpha}f)(y)=\sum_{S\subseteq[n]} \hat{f}(S)(1-2\alpha)^{|S|}\prod_{i\in S} y_i\label{eq:Fourier}.
\end{align}

The following well-known theorem will play an important role in the derivation of Theorem~\ref{thm:main}.
\begin{theorem}[Hypercontractivity Theorem, ~\cite{ODonnellBook,Bonami1970,Beckner75}]
Let $1\leq p < q < \infty$. Then for all $\rho \leq \sqrt{\frac{p-1}{q-1}}$, and all $g: \{-1,1\}^n \mapsto \RR $,
it holds that
\begin{align}
\left[\Expt(\left|(T_\rho g)(X)\right|^q)\right]^\frac{1}{q} \leq \left[\Expt(\left|g(X)\right|^p)\right]^{\frac{1}{p}}.
\end{align}
\label{thm:hyper}
\end{theorem}


\section{Proof of Theorem~\ref{thm:main}}
\label{sec:inter}

In this section we prove Theorem~\ref{thm:main}, which gives a universal upper bound on the mutual information $I(f(X);Y)$, that holds for any \emph{balanced} Boolean function.
The mutual information can be expressed as
\begin{align}
I(f(X);Y)&=H(f(X))-H(f(X)|Y) =1-\mathbb{E}_Y h(P_Y^f).
\end{align}
We note that $h(\cdot)$ admits the following Taylor series
\begin{align}
h\left(\frac{1-p}{2}\right)=1-\sum_{k=1}^{\infty}\frac{\log(e)}{2k(2k-1)}p^{2k},\label{eq:entTaylor}
\end{align}
and can be lower bounded by replacing $p^{2k}$ with $p^{2t}$ for all $k>t$, i.e.,
\begin{align}
h\left(\frac{1-p}{2}\right)&\geq 1-\sum_{k=1}^{t-1}\frac{\log(e)}{2k(2k-1)}p^{2k}-p^{2t}\sum_{k=t}^{\infty}\frac{\log(e)}{2k(2k-1)}\nonumber\\
&=1-\sum_{k=1}^{t-1}\frac{\log(e)}{2k(2k-1)}p^{2k}-\left(1-\sum_{k=1}^{t-1}\frac{\log(e)}{2k(2k-1)}\right)p^{2t}, \label{eq:entTaylor2}
\end{align}
where we have used the fact that $h(0)=0$.
Using~\eqref{eq:entTaylor2}, for any $t\in \mathbb{N}$, we can further lower bound $\mathbb{E}_Y h(P_Y^f)$ as
\begin{align}
&\mathbb{E}_Y h(P_Y^f)=\mathbb{E}_Y h\left(\frac{1-(1-2P_Y^f)}{2}\right)\nonumber\\
&\geq 1-\sum_{k=1}^{t-1}\frac{\log(e)}{2k(2k-1)}\Expt_Y\left((1-2P_Y^f)^{2k}\right)\nonumber\\
&-\left(1-\sum_{k=1}^{t-1}\frac{\log(e)}{2k(2k-1)}\right)\Expt_Y\left((1-2P_Y^f)^{2t}\right).\label{eq:entbound}
\end{align}
Thus, any upper bound on the first $t$ even moments $\Expt[(1-2P_Y^f)^{2k}]$, $k=1,\ldots,t$, would directly translate to an upper bound on $I(f(X);Y)$.

We obtain an upper bound on all even moments of the random variable $1-2P_Y^f$, using a simple trick combined with the Hypercontractivity Theorem.
Formally, we prove the following lemma.

\begin{lemma}\label{lem:momentT}
Let $k\geq 1$ be an integer satisfying $(1-2\alpha)\sqrt{2k-1}\leq 1$. For any balanced Boolean function $f:\{-1,1\}^n\mapsto\{-1,1\}$ we have that
\begin{align}
\mathbb{E}_Y \left((1-2P_Y^{f})^{2k}\right)\leq (2k-1)^{k}(1-2\alpha)^{2k}.\nonumber
\end{align}
\end{lemma}

For the proof, we will need the following proposition.

\begin{proposition}\label{prop:secondmoment}
For any balanced Boolean function $f:\{-1,1\}^n\mapsto\{-1,1\}$ and any $0\leq\rho\leq 1$ we have that
\begin{align}
\mathbb{E}_Y \left((T_{\rho}{f})^{2}(Y)\right)\leq \rho^{2}.\nonumber
\end{align}
\end{proposition}

\begin{proof}
By the definition of the operator $T_{\rho}$:
\begin{align} \label{eq_second_moment}
&\Expt_Y\left((T_{\rho}f)^2(Y)\right)=\mathbb{E}_Y\left(\left(\sum_{S\subseteq[n]} \hat{f}(S)\rho^{|S|}\prod_{i\in S} Y_i \right)^2\right)\nonumber\\
&=\mathbb{E}_Y\bigg(\sum_{S_1\subseteq[n]} \hat{f}(S_1)\rho^{|S_1|}\prod_{i\in S_1} Y_i\sum_{S_2\subseteq[n]}\hat{f}(S_2)\rho^{|S_2|}\prod_{j\in S_2} Y_j\bigg)\nonumber\\
&=\sum_{S_1\subseteq[n]} \hat{f}(S_1)\rho^{|S_1|}\sum_{S_2\subseteq[n]}\hat{f}(S_2)\rho^{|S_2|}\mathbb{E}_Y\left(\prod_{i\in S_1} Y_i \prod_{j\in S_2} Y_j \right)\nonumber\\
&=\sum_{S_1\subseteq[n]} \hat{f}(S_1)\rho^{|S_1|}\sum_{S_2\subseteq[n]}\hat{f}(S_2)\rho^{|S_2|}\mathbf{1}_{S_1=S_2}\nonumber\\
&=\sum_{S\subseteq[n]} \hat{f}^2(S)\rho^{2|S|},
\end{align}
where $\mathbf{1}_{S_1=S_2}$ is the indicator function on the event $S_1=S_2$.
Recalling that $\sum_{S\subseteq[n]}\hat{f}^2(S)=1$ and our assumptions that $\rho\leq 1$ and that \mbox{$\hat{f}(\emptyset)=\mathbb{E}(f(X))=0$}, the ``weight'' assignment $\hat{f}^2(S)$ that maximizes $\sum_{S\subseteq[n]} \hat{f}^2(S)\rho^{2|S|}$ puts all the weight on characters $S$ whose cardinality is $|S|=1$. Hence,
\begin{align}
\sum_{S\subseteq[n]} &\hat{f}^2(S)\rho^{2|S|}\leq \rho^2,\nonumber
\end{align}
as desired.
\end{proof}

\begin{proof}[Proof of Lemma~\ref{lem:momentT}]
Let $f:\{-1,1\}^n\mapsto\{-1,1\}$ be a Boolean function and let $g\triangleq T_{1-2\alpha}f$. Let $k\geq 1$ be an integer, and let $\rho=\sqrt{1/(2k-1)}=\sqrt{(2-1)/(2k-1)}$. By~\eqref{eq:Fourier}, we have
\begin{align}
\Expt_Y((1-2P_Y^f)^{2k})&=\Expt_Y\left(g^{2k}(Y)\right)\nonumber\\
&=\Expt_Y\left(\left(T_{\rho}(T_{1/\rho}g)\right)^{2k}(Y)\right)\nonumber\\
&\leq \left[\Expt_Y\left((T_{1/\rho}g)^2(Y)\right)\right]^{k}\label{eq:hyperapp}\\
&=\left[\Expt_Y\left((T_{(1-2\alpha)\sqrt{2k-1}}f)^2(Y)\right)\right]^{k},\label{eq:gdef}
\end{align}
where~\eqref{eq:hyperapp} follows from the Hypercontractivity Theorem taken with $p=2, q=2k$ (as $\rho$ satisfies the premise of the theorem),
and~\eqref{eq:gdef} from the definition of the function $g$. Invoking Proposition~\ref{prop:secondmoment} with $\rho=(1-2\alpha)\sqrt{2k-1}\leq 1$, we obtain that for any balanced Boolean function and any $k\geq 1$
\begin{align}
&\mathbb{E}_Y \left((1-2P_Y^{f})^{2k}\right)\leq (2k-1)^k(1-2\alpha)^{2k}
\end{align}
as desired.
\end{proof}

The following is an immediate consequence of Lemma~\ref{lem:momentT} and~\eqref{eq:entbound}.

\begin{proposition}\label{prop:mainT}
For any balanced Boolean function $f:\{-1,1\}^n\mapsto\{-1,1\}$, any integer $t\geq 1$ and any
$\tfrac{1}{2}\left( 1-\tfrac{1}{\sqrt{2t-1}}\right)\leq\alpha\leq\tfrac{1}{2}$, we have that
\begin{align}
I(f(X);Y)&\leq \sum_{k=1}^{t-1}\frac{\log(e)}{2k(2k-1)}(2k-1)^{k}(1-2\alpha)^{2k}\nonumber\\
&+\left(1-\sum_{k=1}^{t-1}\frac{\log(e)}{2k(2k-1)}\right)(2t-1)^{t}(1-2\alpha)^{2t}.\label{eq:genBound}
\end{align}
\end{proposition}

Theorem~\ref{thm:main} now follows by evaluating~\eqref{eq:genBound} with $t=2$. Note that for balanced functions the upper bound $(1-2\alpha)^2$, which was the best known bound hitherto, is obtained as a special case of Proposition~\ref{prop:mainT} by setting $t=1$. It is easy to verify that for $\tfrac{1}{3}<\alpha<\tfrac{1}{2}$ the upper bound in Theorem~\ref{thm:main} is tighter. See Figure~\ref{fig:Bounds} for a comparison between the bounds.

%
%

\begin{figure}[htb]
  \centering
      \includegraphics[width=0.9\columnwidth]{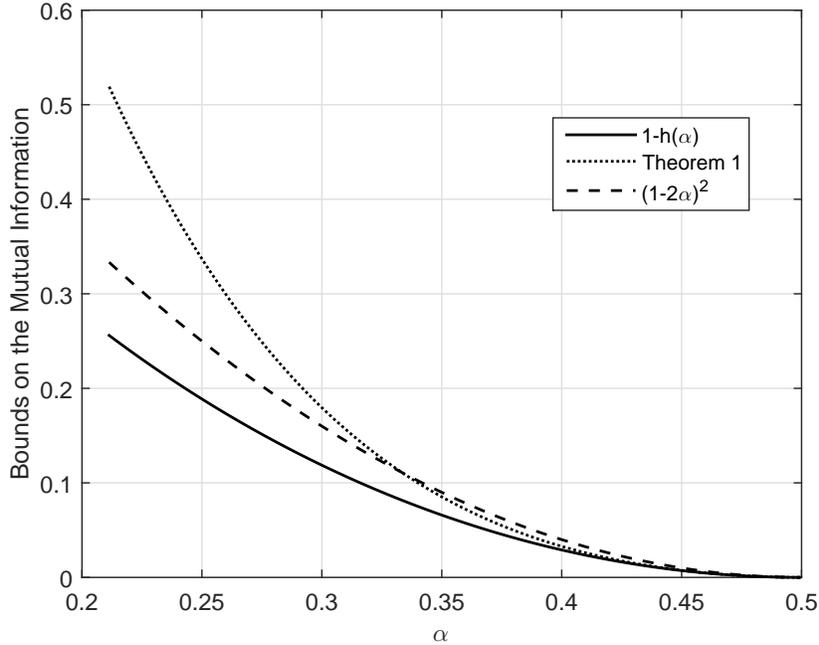}
  \caption{A comparison between the bound from Theorem~\ref{thm:main}, the conjectured bound and best known previous upper bound on $I(f(X);Y)$. Note that the bound from Theorem~\ref{thm:main} is only valid for balanced functions.}
  \label{fig:Bounds}
\end{figure}

\begin{remark}
In~\cite[Appendix B, Remark 6]{ck14}, it is claimed that Conjecture~\ref{conj:ck} can be shown to hold in the limit of $\alpha\to\tfrac{1}{2}$. Theorem~\ref{thm:main} demonstrates this fact for balanced functions, as for $\alpha\to\tfrac{1}{2}$ the ratio between the RHS of~\eqref{eq:mainbound} and the conjectured bound tends to $1$. As we discuss below, a slightly stronger statement can be shown to hold.
\label{rem:courtadepaper}
\end{remark}

The next simple corollary of Theorem~\ref{thm:main} establishes the optimality of the dictatorship function among all balanced functions in the very noisy regime. The proof is essentially a consequence of the discreteness of the space of Boolean functions.

\begin{corollary}[Dictatorship is optimal for very noisy channels]
\label{cor:low_snr}
Let $f:\{0,1\}^n\mapsto\{0,1\}$ be any balanced Boolean function. There exists $\overline{\alpha}_n>0$ such that $I(f(X);Y)\leq1-h(\alpha)$ for all $\alpha\in\left[\tfrac{1}{2}-\overline{\alpha}_n,\tfrac{1}{2}\right]$. In particular, dictatorship is the most informative balanced function in the noise interval $\alpha\in\left[\tfrac{1}{2}-\overline{\alpha}_n,\tfrac{1}{2}\right]$.
\end{corollary}

\begin{proof}
By equation \eqref{eq:entbound} applied with $t=2$, we have that for any balanced Boolean function (and any $\alpha$),
\begin{align}
& I\left(f(X);Y\right)\leq \frac{\log(e)}{2}\Expt_Y\left((1-2P_Y^f)^{2}\right)\nonumber
+ \left(1- \frac{\log(e)}{2}\right)\Expt_Y\left((1-2P_Y^f)^{4}\right). 	\label{eq_low_snr_1}
\end{align}
Lemma \ref{lem:momentT} (applied with $k=2$) implies that for $\tfrac{1}{2}\left(1-\tfrac{1}{\sqrt{3}}\right)\leq \alpha\leq \tfrac{1}{2}$
\begin{align}
\mathbb{E}_Y \left((1-2P_Y^{f})^{4}\right)\leq (2\cdot 2-1)^{2}(1-2\alpha)^{4} = 9(1-2\alpha)^{4}.
\end{align}
Furthermore, by equations \eqref{eq:Fourier} and \eqref{eq_second_moment},  we have that for $\alpha$ in this range
\begin{align}
& \Expt_Y\left((1-2P_Y^f)^{2}\right) = \sum_{S\subseteq [n]} \hat{f}^2(S)(1-2\alpha)^{2|S|} \nonumber \\
& \leq \left(\sum_{|S|=1} \hat{f}^2(S)\right)(1-2\alpha)^{2} + \left(\sum_{|S|\geq 2} \hat{f}^2(S)\right)(1-2\alpha)^{4},
\label{eq_low_snr_3}
\end{align}
where we have used the fact that $\hat{f}(\emptyset)=0$ for balanced functions (Proposition~\ref{prop:fourier_properties}).
Combining the above three inequalities and using the fact that $\sum_{S\subseteq[n]}\hat{f}^2(S)=1$, yields
\begin{align}
& I\left(f(X);Y\right)\leq \frac{\log(e)}{2}\left(1-\sum_{|S|\geq 2} \hat{f}^2(S)\right)(1-2\alpha)^{2} + \nonumber \\
 & + \left(9\left(1- \frac{\log(e)}{2}\right) + \frac{\log(e)}{2}\sum_{|S|\geq 2} \hat{f}^2(S)\right)(1-2\alpha)^{4}.
\label{eq_low_snr_4}
\end{align}

Now, suppose that $f$ is a balanced Boolean function which is \emph{not} a dictatorship function. This implies (in fact, equivalent to)
$\sum_{|S|\geq 2} \hat{f}^2(S) > 0$. By Proposition \ref{prop:fourier_properties}, it therefore must be the case that
\[ \sum_{|S|\geq 2} \hat{f}^2(S) \geq 2^{-2n} = 4^{-n}. \]
Therefore, for such $f$, the RHS of \eqref{eq_low_snr_4} can be upper bounded by
\begin{align}
\frac{\log(e)}{2}(1-4^{-n})(1-2\alpha)^{2}  + \left(9\left(1- \frac{\log(e)}{2}\right) + \frac{\log(e)}{2}4^{-n}\right)(1-2\alpha)^{4}.
\label{eq_low_snr_5}
\end{align}

It can be directly verified that for any $\alpha\in \left[\tfrac{1}{2}-\overline{\alpha}_n, \tfrac{1}{2}\right]$, where \mbox{$\overline{\alpha}_n $ $\triangleq \frac{1}{4}\cdot 2^{-n}$}, the expression in \eqref{eq_low_snr_5} is smaller than
$\frac{\log(e)}{2}(1-2\alpha)^2 < 1-h(\alpha)$, and thus by \eqref{eq_low_snr_4}, for such $\alpha$ we have
\begin{align}
& I\left(f(X);Y\right)\leq 1-h(\alpha),
\end{align}
which completes the proof.
\end{proof}

\begin{remark}
In an unpublished work, Sushant Sachdeva, Alex Samorodnitsky and Ido Shahaf have shown, using different techniques, that dictatorship is optimal among all (not just balanced) Boolean functions from $\{0,1\}^n$ to $\{0,1\}$ for all $\alpha\in\left[\tfrac{1}{2}-2^{-O(n)},\tfrac{1}{2} \right]$.
\label{rem:alexpaper}
\end{remark}

\section{Discussion}

%
%
%
%

A natural question is: What are the limits of the approach pursued in this paper?
To this end we note the following two limitations:

Firstly, the dictatorship function, which is conjectured to be optimal, satisfies $\mathbb{E}_Y \left((1-2P_Y^{f})^{2k}\right)=(1-2\alpha)^{2k}$ for every $k\in\mathbb{N}$. The ratio between the bound in Lemma \ref{lem:momentT} on the $k$-th moment of any balanced Boolean function and $(1-2\alpha)^{2k}$ grows rapidly with $k$.
For this reason, we only get mileage from applying the lemma with $k=1,2$ and not for higher moments.

The second limitation is that Lemma \ref{lem:momentT} upper bounds \emph{each moment separately}, while we are seeking an upper bound
on the entire distribution (weighted sum) of the moments: Quantifying the tradeoff between higher and lower moments seems to be one of the ``brick walls" in 
proving the conjecture. For example, the dictatorship function has the largest second moment among all balanced functions,
but it is not hard to see that the majority function, for example, has a much larger (relatively speaking) $k$th moment for very large values of $k$. To see this, note that for $k\gg 2^n$,
\begin{align}
\mathbb{E}_Y \left((1-2P_Y^{f})^{2k}\right) \gtrsim 2^{-n}\cdot \max_{y}|1-2P_y^{f}|^{2k}.\nonumber
\end{align}
For the majority function, the maximum is attained at $y=(1,1,,1\ldots, 1)$ for which $P_y^{f} \approx 2^{-nD\left(\tfrac{1}{2}||\alpha\right)}$ and consequently $\max_{y}|1-2P_y^{Maj}|\approx 1-2^{-nD\left(\tfrac{1}{2}||\alpha\right)}$. For dictatorship, on the other hand, $|1-2P_Y^f|=1-2\alpha$ for every $y$, and therefore
\begin{align}
\max_{y}|1-2P_y^{Maj}|^{2k}\gg\max_{y}|1-2P_y^{Dict}|^{2k}.\nonumber
\end{align}
Therefore, one cannot hope to prove that there is a single function that simultaneously maximizes all moments; rather, the conjecture postulates
that there is some tradeoff between these values and the largest mutual information is attained by functions that maximize lower moments at the expense of  higher ones.

\section{Acknowledgement}

In a previous version of this paper (http://arxiv.org/abs/1505.05794v1), we have presented Corollary~\ref{cor:low_snr} as a new result, and in addition we have proved that dictatorship is the optimal function for all $\alpha\in\left[0,\tfrac{2^{-2n}}{16n^2} \right]$. We are grateful to Thomas Courtade for bringing to our attention that (slightly weaker versions of) these results were already known, as partially discussed in Remark~\ref{rem:courtadepaper}.


%
\bibliographystyle{IEEEtran}
\bibliography{OrBib2}

%
%

\end{document}